\documentclass[11pt,letterpaper]{article}
\usepackage[margin=1in]{geometry}
\usepackage{amsmath,amsthm,amsfonts,amssymb,amscd}
\usepackage{lastpage}
\usepackage{enumerate}
\usepackage{fancyhdr}
\usepackage{mathrsfs}
\usepackage{xcolor}
\usepackage{graphicx}
\usepackage{listings}
\usepackage[colorlinks=true, urlcolor=blue, linkcolor=blue, citecolor=magenta]{hyperref}
\usepackage{cite}
\usepackage{multirow}
\usepackage{comment}
\usepackage{mathpazo}
\usepackage{indentfirst}
\usepackage{todonotes}
\usepackage{cleveref}
\usepackage{thm-restate}

\allowdisplaybreaks

\hypersetup{%
  colorlinks=true,
  linkcolor=red,
  linkbordercolor={0 0 1}
}

\lstdefinestyle{Matlab}{
    language        = matlab,
    frame           = lines, 
    basicstyle      = \footnotesize,
    keywordstyle    = \color{blue},
    stringstyle     = \color{green},
    commentstyle    = \color{red}\ttfamily
}
\newtheorem{theorem}{Theorem}
\newtheorem{definition}[theorem]{Definition}

\newtheorem{remark}[theorem]{Remark}

\newtheorem{corollary}[theorem]{Corollary}

\newcommand{\xr}[1]{{\color{orange} {\small [Xuandi: #1]}}}

\newcommand{\val}{\text{val}}

\newcommand{\eps}{\varepsilon}

\newcommand{\sfv}{\mathsf{v}}

\newcommand{\ini}{\textsf{ini}}
\newcommand{\tar}{\textsf{tar}}

\newcommand{\kewen}[1]{{\color{brown}Kewen: #1}}

\newcommand{\RIH}[3]{\textrm{Gap}_{1,#1}~{#2}\textrm{-CSP}_{#3}~\textrm{Reconfiguration}}

\title{PCPP-Based Reconfiguration Inapproximability: \\ Query Complexity vs. Soundness Gap Trade-offs}

\author{
Venkatesan Guruswami\thanks{Email: {\tt venkatg@berkeley.edu}. Simons Institute for the Theory of Computing, and Departments of EECS and Mathematics, UC Berkeley. Research supported in part by NSF grants CCF-2228287 and CCF-2211972 and a Simons Investigator award.} 
\and
Xuandi Ren\thanks{Email: {\tt xuandi\_ren@berkeley.edu}. Department of EECS, UC Berkeley. Research supported in part by NSF CCF-2228287.}
\and 
Kewen Wu\thanks{Email: {\tt shlw\_kevin@hotmail.com}. Department of EECS, UC Berkeley. Research supported by a Sloan Research Fellowship and NSF CAREER Award CCF-2145474.}
}

\date{}

\begin{document}

\maketitle
\thispagestyle{empty}


\begin{abstract}
The Reconfiguration Inapproximability Hypothesis (RIH), recently established by Hirahara-Ohsaka (STOC'24) and Karthik-Manurangsi (ECCC'24), studies the hardness of reconfiguring one solution into another in constraint satisfaction problems (CSP) when restricted to approximate intermediate solutions.
In this work, we make a tighter connection between RIH's soundness gap and that of probabilistically checkable proofs of proximity (PCPP). Consequently, we achieve an improved trade-off between soundness and query complexity in Gap CSP Reconfiguration. Our approach leverages a parallelization framework, which also appears in some recent parameterized inapproximability results.
\end{abstract}

\section{Introduction}\label{sec:intro}

Tasks like solving puzzles and planning robot movements can be understood as transformations between different configurations. A valid configuration may be, for example, a feasible arrangement of puzzle pieces or a safe robot posture relative to obstacles.
A single \emph{reconfiguration} step can be interpreted as an adjacency relation between valid configurations, induced by one puzzle move or one robot action.
The study of reconfiguration focuses on the complexity of sequentially transforming an initial configuration to a target configuration for the problem in question.

Formally, let $\Phi$ be an instance of constraint satisfaction problem (CSP) on $n$ variables over alphabet $\Sigma$. 
An assignment $\sigma\in\Sigma^n$ is called a solution if it satisfies all constraints in $\Phi$.
Given two solutions $\sigma,\sigma'\in\Sigma^n$, the goal of the reconfiguration problem is to determine the feasibility of transforming $\sigma$ into $\sigma'$ by sequentially modifying symbols and staying in the solution space. That is, if there exist some positive integer $t$ and assignments $\sigma^{(0)},\sigma^{(1)},\ldots,\sigma^{(t)}\in\Sigma^n$ such that
\begin{enumerate}
\item\label{itm:intro_1} $\sigma^{(0)}=\sigma$, $\sigma^{(t)}=\sigma'$, and each $\sigma^{(i)},\sigma^{(i-1)}$ differ on at most one coordinate;
\item\label{itm:intro_2} each $\sigma^{(i)}$ is a solution to $\Phi$.
\end{enumerate}
Numerous algorithms and hardness results have been established for special cases of $\Phi$.
Below are several representative cases:

\begin{itemize}\setlength\itemsep{0em}
\item Say $\Phi$ is an instance of {\sc 3-SAT}, the satisfiability of which is already NP-complete. Its reconfiguration is PSPACE-complete \cite{gopalan2009connectivity}.
\item Say $\Phi$ is an instance of the {\sc 3-Coloring} problem, another standard NP-complete problem. Its reconfiguration is in P \cite{cereceda2011finding}.
\item Say $\Phi$ is an instance of {\sc Shortest-Path}, which has many textbook efficient algorithms.
However, its reconfiguration is PSPACE-complete \cite{bonsma2013complexity}.
\end{itemize}
We refer interested readers to the surveys by van den Heuvel \cite{van2013complexity} and Nishimura \cite{nishimura2018introduction} for more references.

A meaningful relaxation to the reconfiguration problem above is to allow \emph{approximation} for \Cref{itm:intro_2}: each $\sigma^{(i)}$ only needs to satisfy a small fraction of constraints in $\Phi$.
Indeed, under such relaxation, some previously hard reconfiguration problems become tractable \cite{ito2011complexity,ito2014approximability,ohsaka2022reconfiguration}.
On the hardness side, using the classic PCP theorem \cite{arora1998proof,arora1998probabilistic}, some reconfiguration problems, including {\sc 3-SAT}, were shown to remain NP-hard under approximation \cite{ito2011complexity}.
As the exact version of the {\sc 3-SAT} reconfiguration is PSPACE-complete \cite{gopalan2009connectivity}, there is a huge gap between the NP-hard lower bound \cite{ito2011complexity} and the PSPACE upper bound.
To remedy this, Ohsaka \cite{ohsaka2023gap} put forth a reconfiguration analogue of the PCP theorem, named \emph{Reconfiguration Inapproximability Hypothesis (RIH)}, which implies the PSPACE-completeness of approximate {\sc 3-SAT} reconfiguration, as well as the approximate version of many other popular reconfiguration problems. 

To describe RIH, we first formally define the approximate reconfiguration problem.

\begin{definition}[$\RIH{\eps}{q}{\Sigma}$]\label{def:intro_gapreconfig}
Let $q\ge1$ be an integer, $\eps\in (0,1]$, and $\Sigma$ be a finite alphabet.
An instance of $\RIH{\eps}{q}{\Sigma}$ is specified by a $q$-CSP$_\Sigma$ instance $\Phi$\footnote{See \Cref{def:csp} for the formal definition of $q$-CSP$_\Sigma$.} and two solutions $\sigma,\sigma'\in\Sigma^n$. 
The task is to distinguish the following two cases:
\begin{itemize}
    \item \textsc{Yes Case.} 
    There exist assignments $\sigma^{(0)},\ldots,\sigma^{(t)}\in\Sigma^n$ such that $\sigma^{(0)}=\sigma$, $\sigma^{(t)}=\sigma'$, each $\sigma^{(i)},\sigma^{(i-1)}$ differ on at most one coordinate, and each $\sigma^{(i)}$ satisfies all constraints in $\Phi$.
    \item \textsc{No Case.}
    There exist no assignments $\sigma^{(0)},\ldots,\sigma^{(t)}\in\Sigma^n$ such that $\sigma^{(0)}=\sigma$, $\sigma^{(t)}=\sigma'$, each $\sigma^{(i)},\sigma^{(i-1)}$ differ on at most one coordinate, and each $\sigma^{(i)}$ satisfies at least $\eps$ fraction of the constraints in $\Phi$.
\end{itemize}
\end{definition}

In light of \Cref{def:intro_gapreconfig}, RIH postulates that, for some fixed $q,\Sigma$, and $\eps < 1$, it is PSPACE-hard\footnote{By brute-force, $\RIH\eps q\Sigma$ is in PSPACE. Hence this in fact asserts PSPACE-completeness.} to solve $\RIH{\eps}{q}{\Sigma}$.
Very recently, Hirahara-Ohsaka \cite{HO24} and Karthik-Manurangsi \cite{KM24} independently proved RIH.
Using additional reductions by Ohsaka \cite{Ohs24}, the state-of-the-art quantitative bound is the PSPACE-hardness of $\RIH{0.9942}{2}{\Sigma}$ for some constant-sized alphabet $\Sigma$.\footnote{\cite{Ohs24ICALP} shows the PSPACE-hardness of $\RIH{1-\varepsilon}{2}{\Sigma}$ for $\varepsilon$ and $\Sigma$ both being explicit constants, namely, $\varepsilon> 10^{-18}$ and $|\Sigma|<2\times10^6$. Furthermore, as shown by \cite{KM24,HO24}, if the arity $q$ and the alphabet size $|\Sigma|$ are arbitrarily large constants, then the soundness can be arbitrarily small. In this work, we focus on minimizing the soundness of Gap 2-CSP$_{\Sigma}$ Reconfiguration, while keeping $|\Sigma|$ a constant.}
We remark the $0.0058$ soundness gap here will become even worse in downstream applications \cite{ohsaka2023gap}; and it is unclear if the gap amplification techniques in \cite{Ohs24} can be improved to give significantly better bounds.

The main contribution of our work is a tighter connection between the soundness gap of RIH and the one of probabilistic checkable proof of proximity (PCPP), formulated in \Cref{thm:main}.
We will introduce PCPP shortly. At this point, let us just comment that PCPP is a widely studied concept in complexity theory \cite{BGHSV06, dinur2006assignment,Harsha2004Robust}.

\begin{restatable}[Main]{theorem}{thmmain}
\label{thm:main}
    For some fixed constant $\delta>0$, if for every Boolean circuit of size $n$, there is a $q$-query PCPP construction with proximity parameter $\delta$, soundness $1-\varepsilon$ and randomness $O(\log n)$, then for some alphabet $\Sigma$ of constant size, $\RIH{1-\eps}{(q+1)}{\Sigma}$ is PSPACE-hard.
\end{restatable}

In other words, up to an extra query (arity of CSP), the parameters of Gap CSP Reconfiguration and PCPPs are identical. Our work thus connects RIH to PCPP in a quantitatively strong way.
We remark that \cite{KM24} implicitly showed, starting from a PCPP with the above parameters, 
$\RIH{1-\frac\eps4}{(q+1)}{\Sigma}$ is PSPACE-hard, which is similar to \Cref{thm:main} but has inferior parameters.
Indeed, our analysis builds on their proof and our tighter reduction relies on a parallelization trick, inspired by recent works in parameterized hardness of approximation \cite{LRSW23,GLRSW24,guruswami2024almost}. We remark that the idea of parallelization is fairly general, and it works as long as there are multiple PCPPs with the same structure of queries used by the verifier. Fortunately, this turns out to be the case for the proof of RIH in \cite{KM24}.


To get a result for Gap $2$-CSP Reconfiguration, we note that \cite[Lemma 5.4]{Ohs24} gives a way to trade the arity of the CSP with the soundness gap. Combined with \Cref{thm:main}, this leads to the following corollary. 

\begin{restatable}{corollary}{cormain}
\label{cor:main}
    For some fixed constant $\delta>0$, if for every Boolean circuit of size $n$, there is a $q$-query PCPP construction with proximity parameter $\delta$, soundness $1-\varepsilon$ and randomness $O(\log n)$, then for some alphabet $\Sigma$ of constant size, $\RIH{1-\frac\eps{q+1}}{2}{\Sigma}$ is PSPACE-hard.
\end{restatable}


Due to the lack of PCPP constructions with explicit constant query complexity and soundness, we are not yet able to use \Cref{cor:main} to improve the 0.9942 soundness for Gap 2-CSP Reconfiguration in \cite{Ohs24}. 
We view our \Cref{thm:main} and \Cref{cor:main} as another concrete motivation to understand the soundness-query tradeoff of PCPPs.

Finally, we briefly explain the PCPPs needed in \Cref{thm:main} and defer its formal definition to \Cref{def:pcpp}. 
Let $\mathcal C\colon\{0,1\}^n\to\{0,1\}$ be a Boolean function, given by its circuit representation.
Intuitively, the PCPP is a verifier $\mathcal V$ that efficiently decides, given access to an auxiliary proof $\pi\in\{0,1\}^*$ 
\footnote{Our reduction also works if the PCPP has a constant-sized proof alphabet other than $\{0,1\}$. For ease of presentation and consistency, we assume the PCPPs use binary proofs.},
if a particular input $x\in\{0,1\}^n$ is a solution of $\mathcal C$ or $x$ is far from any solution of $\mathcal C$.
The \emph{query complexity} is the maximum number of bits that $\mathcal V$ reads on $(x,\pi)$; the \emph{proximity parameter} is the distance of $x$ to solutions; the \emph{soundness} is the acceptance probability given $x$ is far from solutions; and the \emph{randomness} is the number of random coins used in $\mathcal V$.
PCPP is implicit in the low-degree testers \cite{arora1998probabilistic,arora1998proof} that utilize auxiliary oracles to test if a given function is close to a low-degree polynomial. It is then explicitly defined in \cite{BGHSV06}, and concurrently in \cite{dinur2006assignment} with a different name Assignment Tester. We refer interested readers to Harsha's thesis \cite{Harsha2004Robust} for an overview of PCPP's history and the use of PCPP to reprove the PCP theorem with a clean modular version of proof composition.



While the PCP theorem admits verifiers of subconstant soundness and two queries~\cite{moshkovitz2008two}, we are not aware of 
explicit tradeoffs between soundness and query complexity for PCPPs.


\paragraph{Organization of the Paper.}
In \Cref{sec:prelim}, we define necessary notations.
In \Cref{sec:construction}, we describe our construction and give a formal proof of \Cref{thm:main}.
\section{Preliminaries}\label{sec:prelim}

We use $[n]$ to denote the set $\{1,2,\ldots,n\}$ for any positive integer $n$.

Below we first give the definition of Probabilistic Checkable Proof of Proximity (PCPP).

\begin{definition}[Probabilistic Checkable Proof of Proximity (PCPP)]\label{def:pcpp}
    For integers $n,m,r,q \in \mathbb N$, an $(n,m,r,q)$-PCPP verifier for a circuit $\mathcal C:\{0,1\}^n \to \{0,1\}$ with proximity parameter $\delta>0$ and soundness $\kappa>0$ is a probabilistic polynomial-time algorithm\footnote{In traditional definitions $\mathcal C$ is also given as the input to $\mathcal V$, but for the ease of defining parallelization later, we define $\mathcal V$ to be the verifier \textit{after} seeing $\mathcal C$ as the input.} $\mathcal V$ such that given oracle access to an assignment $x \in \{0,1\}^n$ and an auxiliary proof $\pi \in \{0,1\}^m$, the following holds:
    \begin{itemize}
        \item $\mathcal V$ tosses $r$ random coins, then makes $q$ queries on $x \circ \pi$ and decides to accept or reject.
        \item If $x$ is a solution of $\mathcal C$, then there exists some $\pi$ such that $\mathcal V$ always accepts.
        \item If $x$ is $\delta$-far from any solution of $\mathcal C$, then for every $\pi$, $\mathcal V$ accepts $\pi$ with probability at most $\kappa$.
    \end{itemize}
\end{definition}

\begin{definition}[$q$-CSP$_{\Sigma}$]
\label{def:csp}
    A $q$-CSP$_{\Sigma}$ instance is a tuple $\Pi=(G=(V,E),\Sigma,C=\{c_e\}_{e \in E})$, where:
    \begin{itemize}
        \item $V$ is the set of variables.
        \item $\Sigma$ is the alphabet of each variable.
        \item $E$ and $C$ together define the constraints. Specifically, $E$ is a set of $q$-hyperedges on $V$, and for each $e=(u_1,\ldots,u_q)\in E$, there is a function $c_e:\Sigma^q \to \{0,1\}$ in $C$ describing the constraint. 
    \end{itemize}
    An assignment $\psi$ is a function from $V$ to $\Sigma$. The value of $\Pi$ under assignment $\psi$ is
    $$\val_{\Pi}(\psi)=\frac{1}{|E|} \sum_{e=(u_1,\ldots,u_q) \in E} c_e(\psi(u_1), \ldots, \psi(u_k)).$$
\end{definition}

\begin{definition}[Error Correcting Codes]
A binary error correcting code $C$ of message length $k$ and block length $n$ is a mapping from $\{0,1\}^k$ to $\{0,1\}^n$, such that for any different messages $m_1,m_2 \in \{0,1\}^k$, $C(m_1)$ and $C(m_2)$ have Hamming distance at least $d_C$ for some $d_C>0$ defined as the \textit{distance} of the code.

\end{definition}

We refer to \Cref{def:intro_gapreconfig} for the definition of $\RIH{\eps}{q}{\Sigma}$.
\section{Main Construction}\label{sec:construction}

In this section, we first review the proof in \cite{KM24} and introduce the idea of parallelization. Then we show how to apply parallelization to \cite{KM24}'s construction to get a better tradeoff between query complexity and soundness.

\subsection{Proof Overview}\label{sec:overview}


\paragraph{An Overview of \cite{KM24}'s Proof.}
The starting point is $\RIH{1}{2}{\Sigma}$ problem for some constant-sized $\Sigma$, which is known to be PSPACE-hard \cite{gopalan2009connectivity,ohsaka2023gap}. Given such a 2-CSP$_{\Sigma}$ instance $\Pi$ on $n$ variables, they take a binary error correcting code $C:\{0,1\}^{n \log |\Sigma|} \to \{0,1\}^{m}$ to encode assignments of $\Pi$. 
To create a gap, they first construct $4m$ many binary variables, written as $x_1,\ldots,x_4 \in \{0,1\}^{m}$, where each $x_i$ is supposed to be the encoding of some assignment $y_i \in \Sigma^n$ of $\Pi$\footnote{The significance of 4 copies is that, after dropping any copy, it's still possible to do a majority vote among the rest 3 assignments.}. 
Then, they build 4 PCPPs $\mathcal V_1,\ldots,\mathcal V_4$, where each $\mathcal V_i$ queries the three strings $\{x_j\}_{j \in [4] \setminus \{i\}}$ and an auxiliary proof $\pi_i\in\{0,1\}^{m'}$, then verifies the following:
\begin{itemize}
    \item For each $j \in [4] \setminus \{i\}$, the assignment $y_j$ decoded from $x_j$ satisfies $\Pi$.
    \item For each $j_1\neq j_2 \in [4] \setminus \{i\}$, $y_{j_1}$ and $y_{j_2}$ differ in at most one bit. Here we write $y_{j_1},y_{j_2}$ as binary vectors in $\{0,1\}^{n \log |\Sigma|}$. The binary alphabet ensures that, among the three strings $\{y_j\}_{j \in [4] \setminus \{i\}}$ which have pairwise distance at most 1, two of them must be equal. This helps to do a majority vote of $\{y_j\}_{j \in [4] \setminus \{i\}}$, as we will see in the soundness analysis.
\end{itemize}

By the definition of PCPPs, if the above two conditions hold, then there is a proof $\pi_i$ which always makes $\mathcal V_i$ accept; otherwise, for every proof $\pi_i$, $\mathcal V_i$ accepts with probability at most $1-\varepsilon$. Their final CSP construction is the combination of the 4 PCPPs, together with a variable $\sfv\in[4]$ indicating which PCPP is currently being checked. Specifically, for each $q$-ary constraint $\varphi$ in some $\mathcal V_i$, they add $\sfv$ to $\varphi$ so it becomes a $(q+1)$-ary constraint $\varphi'$. 
Then $\varphi'$ is satisfied iff $\sfv\neq i$ (i.e., we are not checking $\mathcal V_i$),\footnote{We slightly overload notation by using $\sfv$ to also denote the value it takes.} or $\sfv=i$ and $\varphi$ is satisfied (i.e., we are checking $\mathcal V_i$ and the check is passed).

For the completeness, they rely on the fact that at any time we only visit 3 $x_j$'s and 1 $\pi_j$, namely, $\{x_j\}_{j \in [4] \setminus \{\sfv\}}$ and $\pi_{\sfv}$. 
This gives us the flexibility to change $x_{\sfv}$ to an adjacent assignment (and $\{\pi_j\}_{j \in [4] \setminus \{\sfv\}}$ accordingly) while keeping the instance fully satisfied. 

For the soundness, suppose we have a reconfiguration sequence $\sigma^{(0)},\ldots,\sigma^{(t)}$ of the final instance, where each $\sigma^{(k)}$ consists of $\sfv^{(k)},\{x_j^{(k)}\}_{j\in[4]},\{\pi_j^{(k)}\}_{j\in[4]}$.
Since every $\sigma^{(k)}$ has value at least $1-\varepsilon$, we are able to extract an assignment\footnote{While the construction of $z^{(k)}$ is not important for our purpose, it is not hard to describe: $z^{(k)}$ is the majority vote of the three assignments $\{y_j^{(k)}\}_{j \in [4] \setminus \{\sfv^{(k)}\}}$, where $y_j^{(k)} \in \{0,1\}^{n\log|\Sigma|}$ is decoded from $x_j^{(k)}$. The distance property is guaranteed by the PCPP verifier $\mathcal V_{\sfv^{(k)}}$.} $z^{(k)}\in\{0,1\}^{n\log|\Sigma|}$ from $\sigma^{(k)}$, such that every $z^{(k)}$ fully satisfies $\Pi$ and differs with $z^{(k+1)}$ in at most one coordinate, which is a valid reconfiguration sequence for the original problem $\Pi$.

In the above construction, a $q$-query $(1-\eps)$-soundness PCPP leads\footnote{The proximity parameter of the PCPP is not very important: it just needs to be smaller than half of the relative distance of the error correcting code for unique decoding.} to the PSPACE-hardness of $\RIH{1-\frac{\varepsilon}{4}}{(q+1)}{}$, where the one extra query is on $\sfv$, and the soundness $1-\frac{\varepsilon}{4}$ comes from putting together the 4 PCPPs and letting 3 of them pass for free.

\paragraph{Parallelization.}
Now we discuss our tweak on top of the above analysis.
The key observation is that the 4 PCPPs in \cite{KM24}'s construction have the same structure. In other words, they are the same verifiers, albeit applied on different strings. Specifically, for each randomness $\omega \in \{0,1\}^r$, there is a subset of indices $I_\omega$, such that for each $i \in [4]$, the PCPP verifier $\mathcal V_i$ queries the subset $I_\omega$ on the string $x_{j_1} \circ x_{j_2} \circ x_{j_3} \circ \pi_i$, where $\{j_1,j_2,j_3\}=[4] \setminus \{i\}$.

Given this observation, we ``parallelize'' the 4 PCPPs in a natural way. We stack $x_1,\ldots,x_4\in \{0,1\}^{m}$ (and $\pi_1,\ldots,\pi_4 \in \{0,1\}^{m'}$ respectively) into 4 layers, and interpret it as a string $x$ of length $m$ (a string $\pi$ of length $m'$, respectively) over alphabet $\{0,1\}^4$. By querying a subset $I$ in $(x \circ \pi)$, we get the information on $(x_j \circ \pi_j)_I$ for every $j \in [4]$. 
After that, we perform the test in $\mathcal V_i$, which involves $x_{j_1}\circ x_{j_2} \circ x_{j_3} \circ \pi_i$ where $\{j_1,j_2,j_3\}=[4] \setminus \{i\}$. 
Suppose there is a PCPP construction with arity $q$ and soundness $1-\varepsilon$, applying our parallelization trick on \cite{KM24}'s construction results in the PSPACE-hardness of $\RIH{1-\varepsilon}{(q+1)}{}$ over the alphabet $\{0,1\}^4$, where the one extra arity comes from $\sfv$ as in \cite{KM24}. 

This idea of parallelization is inspired by recent works in parameterized hardness of approximation \cite{LRSW23,GLRSW24,guruswami2024almost}. In the world of parameterized complexity, we often encounter problems involving few variables but large alphabet size. 
The large alphabet size prevents direct applications of standard tools from the NP world, which typically requires a constant-sized alphabet. However, sometimes the problem is structured in a way that we can associate $[n]$ as $\Sigma^{\log _{|\Sigma|} n}$ for a constant-sized $\Sigma$, and view the constraints as acting in parallel across the $\log_{|\Sigma|} n$ layers of the variables. This allows us to apply known techniques, such as error correcting codes and PCPPs, in parallel to each layer, and verify the constraints across all layers simultaneously.

\paragraph{Discussion on \cite{HO24}.}

We remark that Hirahara and Ohsaka \cite{HO24} independently and concurrently proved RIH, using a different construction from that of Karthik‑Manurangsi \cite{KM24}. Although PCPP is also used, the verifier in \cite{HO24} consists of several (a constant number of) stages. The query complexity is thus the sum of query complexity in each stage, and the soundness is a bit hard for us to formulate into a clean form due to the case analysis involved.

\subsection{Proof of Theorem \ref{thm:main}}

\begin{definition}[Parallelizable PCPPs]
    For any constant $t\ge 1$, we say a set of $t$ $(n,m,r,q)$-PCPPs $\mathcal V_1,\ldots,\mathcal V_t$ (on circuits $\mathcal C_1,\ldots,\mathcal C_t$ and with inputs $x^{(1)}\circ \pi^{(1)}, \ldots, x^{(t)} \circ \pi^{(t)}$, respectively) are parallelizable if for every randomness in $\{0,1\}^r$, $\mathcal V_1,\ldots,\mathcal V_t$ query the same set of $q$ locations in their respective $x^{(i)} \circ \pi^{(i)}$.
\end{definition}

We think of the input of the $t$ PCPPs as written in a $t \times (n+m)$ table. Henceforth, we slightly overload notation by using $x$ and $\pi$ to denote \textit{unassigned variables} instead of \textit{strings}. Specifically:
\begin{itemize}
    \item We use $x_1,\ldots,x_n,\pi_1,\ldots,\pi_m$ to denote the variables on each column, which are supposed to take values in $\{0,1\}^t$. We abbreviate variables $(x_1, \ldots, x_n)$ and $(\pi_1, \ldots, \pi_m)$ as $x$ and $\pi$, respectively.
    
    \item We use $x^{(1)},\ldots,x^{(t)}, \pi^{(1)},\ldots,\pi^{(t)}$ to denote the variables on each row, where $x^{(i)}$ is supposed to take values in $\{0,1\}^n$ and $\pi^{(i)}$ is supposed to take values in $\{0,1\}^m$.
\end{itemize}



\begin{theorem}
\label{thm:pcpp2csp}
    For any constant $t\ge 1$, let $\mathcal V_1,\ldots,\mathcal V_t$ be $t$ parallelizable $(n,m,r,q)$-PCPPs (with respect to circuits $\mathcal C_1,\ldots,\mathcal C_t:\{0,1\}^n \to \{0,1\}$), there is a $(q+1)$-CSP instance $\Pi$ which can be built in polynomial time, such that:
    \begin{itemize}
        \item \textsc{Variables and Alphabet.} $\Pi$ has $n+m+1$ variables $\{\sfv\} \cup x \cup \pi$ over alphabet $\{0,1\}^t$.
        \item \textsc{Completeness.} Fix an assignment $\psi$ of $\Pi$. If for some $i \in [t]$, $\mathcal V_i$ accepts $\psi(x^{(i)}) \circ \psi(\pi^{(i)})$ with probability 1, then $\val_\Pi(\psi)=1$.
        \item \textsc{Soundness.} Fix an assignment $\psi$ of $\Pi$. If for every $i \in [t]$, $\mathcal V_i$ accepts $\psi(x^{(i)}) \circ \psi(\pi^{(i)})$ with probability at most $\kappa$, then $\val_{\Pi}(\psi) \le \kappa$.
    \end{itemize}
\end{theorem}

\begin{proof}
    $\Pi$ is constructed by simply ``parallelizing'' the $t$ PCPP verifiers. Specifically, for each randomness in $\{0,1\}^r$, let $I$ be the size-$q$ local window that every $\mathcal V_i$ queries on their respective $x \circ \pi$. We add a $(q+1)$-ary constraint $c$ to $\Pi$ as follows:
    \begin{itemize}
        \item $c$ is on variables $\{\sfv\} \cup (x \circ \pi)_I$.
        \item $c$ is satisfied if and only if:
        \begin{itemize}
            \item $\sfv$ takes value in $[t]$ (We treat $\{0,1\}^t$, the alphabet of $\sfv$, as a super-set of $[t]$).
            \item Suppose $\sfv$ takes value $i \in [t]$, then $\mathcal V_i$ accepts the assignment on $(x^{(i)} \circ \pi^{(i)})_I$.
        \end{itemize}
    \end{itemize}
    The completeness and soundness are easy to check.
\end{proof}

We refer to \Cref{fig:1} for an illustration of our $t$-Parallel PCPP construction.

\begin{figure}[htbp]
  \centering
  \includegraphics[width=0.8\textwidth]{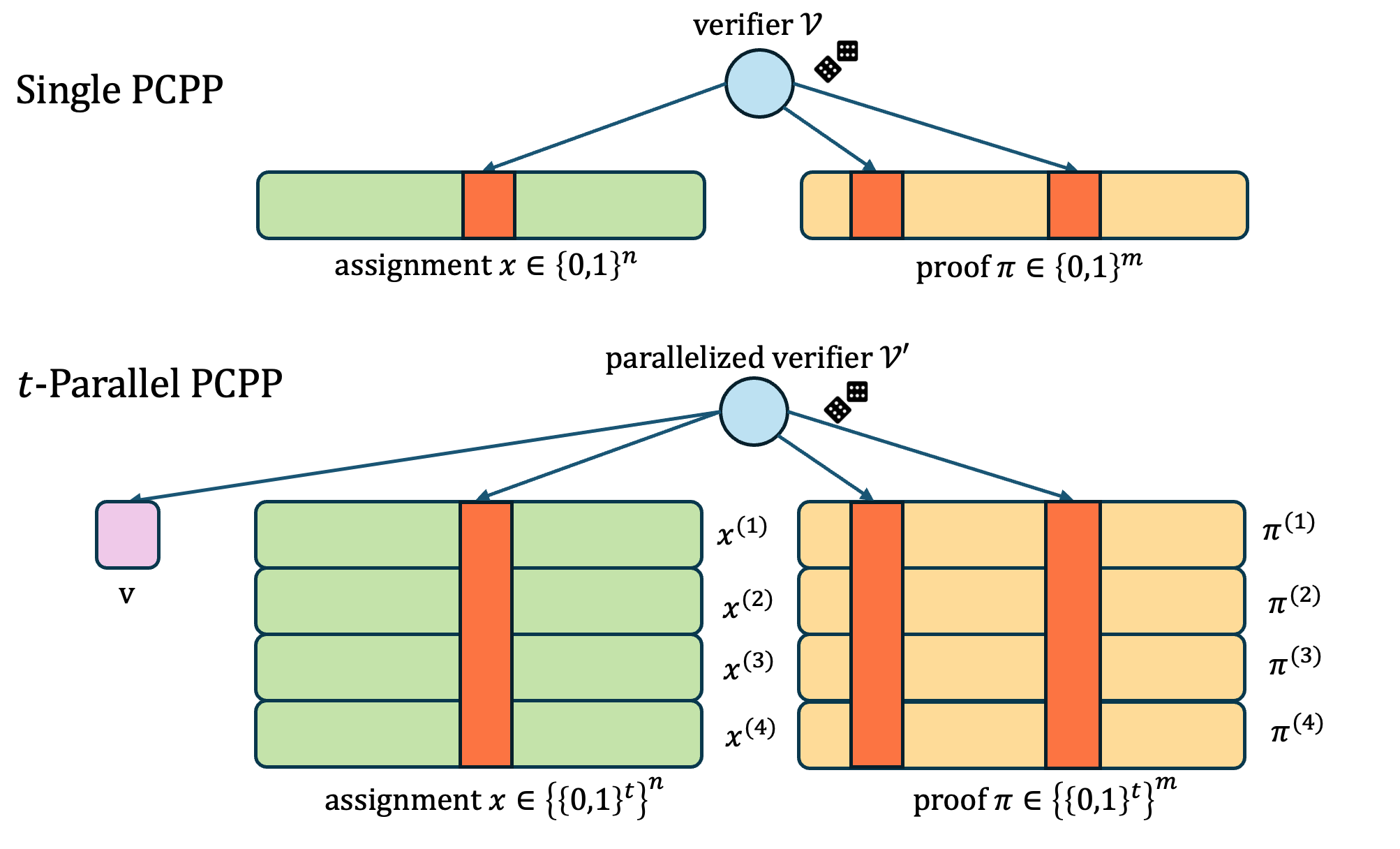}
  \caption{An illustration of the $t$-Parallel PCPP construction ($t=4$).}
  \label{fig:1}
\end{figure}

\begin{definition}[Gap $t$-Parallel $q$-PCPP$_{\Sigma}$ Reconfiguration]
    For any constant $t \ge 1$ and any $0 \le s \le c \le 1$, Gap$_{c,s}$ $t$-Parallel $q$-PCPP Reconfiguration asks, for $t$ parallelizable $(n,m,r,q)$-PCPPs $\mathcal V_1,\ldots,\mathcal V_t$ and their two assignments $\psi^{\ini},\psi^{\tar}:
    (x \cup \pi) \to \{0,1\}^t$, to distinguish between the following two cases:
    \begin{itemize}
        \item \textsc{Yes Case.} There is a reconfiguration sequence $\Psi=(\psi^{\ini},\ldots,\psi^\tar)$, such that for every $\psi \in \Psi$, there exists $i \in [t]$ such that $\mathcal V_i$ accepts $\psi(x^{(i)})\circ \psi(\pi^{(i)})$ with probability at least $c$.
        \item \textsc{No Case.} For any reconfiguration sequence $\Psi=(\psi^\ini,\ldots,\psi^\tar)$, there exists $\psi \in \Psi$, such that for all $i \in [t]$, $\mathcal V_i$ accepts $\psi(x^{(i)}) \circ \psi(\pi^{(i)})$ with probability less than $s$.
    \end{itemize}
\end{definition}

Applying our construction in \Cref{thm:pcpp2csp}
, we have the following corollary:

\begin{corollary}
\label{cor:pcpp-to-csp}
    For any constant $t \ge 1$ and any $0 \le s \le 1$, Gap$_{1,s}$ $t$-Parallel $q$-PCPP Reconfiguration reduces to Gap$_{1,s}$ $(q+1)$-CSP$_\Sigma$ Reconfiguration in polynomial time, where $\Sigma=\{0,1\}^t$.
\end{corollary}

\begin{remark}\label{rmk:pcpp_alphabet}
    The reduction works even if the PCPPs use a constant-sized proof alphabet $\Sigma$ other than $\{0,1\}$. In this case, the resulting Gap$_{1,s}$ $(q+1)$-CSP$_{\Sigma'}$ Reconfiguration has alphabet $\Sigma'=\Sigma^t$. For ease of presentation and consistency, we assume the PCPPs use binary alphabets.
\end{remark}

\begin{theorem}[Implicit in \cite{KM24}]
\label{thm:km24}
    For some fixed constant $\delta>0$, if for every Boolean circuit of size $n$, there is a $q$-query PCPP construction with proximity parameter $\delta$, soundness $1-\varepsilon$ and randomness $O(\log n)$. Then Gap$_{1,1-\varepsilon}$ $4$-Parallel $q$-PCPP$_{\Sigma}$ Reconfiguration is PSPACE-hard.
\end{theorem}

Combining \Cref{thm:km24} with \Cref{cor:pcpp-to-csp}, we have \Cref{thm:main} as a corollary.

\thmmain*




\bibliographystyle{alpha}
\bibliography{main}

\end{document}